\documentclass{article}

\def\noheaderplainsetup{

\topmargin=0pt \headheight=0pt \headsep=0pt  \oddsidemargin=0pt \evensidemargin=0pt  \textheight=8.4truein \textwidth=6.4truein}   

\noheaderplainsetup

\usepackage{amsfonts}

\begin{document}


 \newcommand{\one}{\mbox{\sc One}}
 \newcommand{\two}{\mbox{\sc Two}}
 \newcommand{\three}{\mbox{\sc Three}}
 \newcommand{\four}{\mbox{\sc Four}}
 \newcommand{\fif}{\mbox{\bf CL15}}
\newcommand{\xx}{\wp}
\newcommand{\col}[1]{\mbox{$#1$:}}

\newcommand{\seq}[1]{\langle #1 \rangle}           


\newcommand{\mla}{\mbox{{\Large $\wedge$}}}
\newcommand{\mle}{\mbox{{\Large $\vee$}}}

\newcommand{\pst}{\mbox{\raisebox{-0.01cm}{\scriptsize $\wedge$}\hspace{-4pt}\raisebox{0.16cm}{\tiny $\mid$}\hspace{2pt}}}
\newcommand{\gneg}{\neg}                  
\newcommand{\mli}{\rightarrow}                     
\newcommand{\cla}{\mbox{\large $\forall$}}      
\newcommand{\cle}{\mbox{\large $\exists$}}        
\newcommand{\mld}{\vee}    
\newcommand{\mlc}{\wedge}  
\newcommand{\ade}{\mbox{\Large $\sqcup$}}      
\newcommand{\ada}{\mbox{\Large $\sqcap$}}      
\newcommand{\add}{\sqcup}                      
\newcommand{\adc}{\sqcap}                      

\newcommand{\tlg}{\bot}               
\newcommand{\twg}{\top}               
\newcommand{\st}{\mbox{\raisebox{-0.05cm}{$\circ$}\hspace{-0.13cm}\raisebox{0.16cm}{\tiny $\mid$}\hspace{2pt}}}
\newcommand{\stl}{\st\hspace{-2pt} ^{\mbox{\tiny L}}}
\newcommand{\stt}{\st \hspace{-2pt} ^{\mbox{\tiny T}}}
\newcommand{\costl}{\cost \hspace{-1pt} ^{\mbox{\tiny L}}}
\newcommand{\costt}{\cost \hspace{-1pt} ^{\mbox{\tiny T}}}
\newcommand{\bst}{\mbox{\raisebox{-0.05cm}{$\circ$}\hspace{-0.135cm}\raisebox{0.03cm}{.}\hspace{-0.08cm}\raisebox{0.16cm}{\tiny $\mid$}\hspace{2pt}}}
\newcommand{\cost}{\mbox{\raisebox{0.12cm}{$\circ$}\hspace{-0.13cm}\raisebox{0.02cm}{\tiny $\mid$}\hspace{2pt}}}
\newcommand{\bcost}{\mbox{\raisebox{0.12cm}{$\circ$}\hspace{-0.135cm}\raisebox{0.2cm}{.}\hspace{-0.08cm}\raisebox{0.02cm}{\tiny $\mid$}\hspace{2pt}}}
\newcommand{\pcost}{\mbox{\raisebox{0.12cm}{\scriptsize $\vee$}\hspace{-4pt}\raisebox{0.02cm}{\tiny $\mid$}\hspace{2pt}}}
\newcommand{\uvalid}{\mbox{$\vdash\hspace{-5pt}\vdash\hspace{-5pt}\vdash$}}


\newtheorem{theoremm}{Theorem}[section]
\newtheorem{corollaryy}[theoremm]{Corollary}
\newtheorem{definitionn}[theoremm]{Definition}
\newtheorem{lemmaa}[theoremm]{Lemma}
\newtheorem{propositionn}[theoremm]{Proposition}
\newtheorem{conventionn}[theoremm]{Convention}
\newtheorem{examplee}[theoremm]{Example}
\newtheorem{remarkk}[theoremm]{Remark}
\newtheorem{factt}[theoremm]{Fact}
\newtheorem{exercisee}[theoremm]{Exercise}

\newenvironment{exercise}{\begin{exercisee} \em}{ \end{exercisee}}
\newenvironment{definition}{\begin{definitionn} \em}{ \end{definitionn}}
\newenvironment{theorem}{\begin{theoremm}}{\end{theoremm}}
\newenvironment{lemma}{\begin{lemmaa}}{\end{lemmaa}}
\newenvironment{proposition}{\begin{propositionn} }{\end{propositionn}}
\newenvironment{corollary}{\begin{corollaryy} \em}{\end{corollaryy}}
\newenvironment{remark}{\begin{remarkk} \em}{\end{remarkk}}
\newenvironment{proof}{ {\bf Proof.} }{\  \rule{2.5mm}{2.5mm} \vspace{.2in} }
\newenvironment{idea}{ {\bf Idea.} }{\  \rule{1.5mm}{1.5mm} \vspace{.15in} }
\newenvironment{example}{\begin{examplee} \em}{\end{examplee}}
\newenvironment{fact}{\begin{factt}}{\end{factt}}

\title{A new face of the branching recurrence of computability logic}
\author{Giorgi Japaridze}
\date{}
\maketitle

\begin{abstract} This letter introduces a new, substantially simplified version of the branching recurrence operation of computability logic, and proves its equivalence to the old, ``canonical''  version. 
\end{abstract}


\  


\section{Introduction}
{\em Computability logic} (CoL) is a long-term project for redeveloping logic on the basis of a constructive game semantics. 
The approach induces a rich collection of logical operators, standing for various natural operations on games. Among the most important of those is the {\em branching recurrence} operator $\st$, in its logical behavior reminiscent of Girard's \cite{Gir87} storage operator $!$ and (especially) Blass's \cite{Bla92} repetition operator $R$, yet different from either: for instance, the principle $\cost\st P\mli\st\cost P$ \ ($\cost$ means $\gneg\st\gneg$) is valid in CoL while linear of affine logics do not prove it with $\st,\cost$ understood as $!,?$, and $\mli$ as linear implication; and the principle $P\mlc \st (P\mli P\mlc P)\mlc \st(P\mld P\mli P)\mli \st P$ is not valid in CoL (nor provable in affine logic) while its counterpart is validated by Blass's semantics. 

 Recent years (\cite{Japtocl1}-\cite{Japseq}, \cite{Japtowards}-\cite{Ver} and more) have seen rapid and sustained progress in constructing sound and complete axiomatizations for various, often quite expressive, fragments of CoL, at both the propositional and the first-order levels. Those fragments, however, have typically been recurrence-free,\footnote{The so called intuitionistic fragment of CoL, studied in \cite{Japjsl,Propint,Ver}, is the only exception. There, however, the usage of $\st$ is limited to the very special form/context $\st E\mli F$.} and finding  syntactic descriptions (such as axiomatizations) of the logic induced by $\st$ remains among the greatest challenges in the  CoL enterprise. Among the reasons why the progress towards 
overcoming this challenge has been so slow is the degree of technical involvement of the existing, ``canonical'' definition of $\st$ as given in \cite{Jap03,Japfin}. It has become increasingly evident that replacing that definition by a substantially less intricate counterpart would be necessary in order to make a breakthrough in syntactically taming $\st$. This is exactly  what the present paper is devoted to. It introduces a technically new, very simple and compact, definition of $\st$ and proves that the new version of $\st$ is logically equivalent to the old one. This means that, from now on, studies of the fragments of CoL involving $\st$ can safely focus on the new version of this operator without losing any already known results concerning $\st$ and without any need to reintroduce or revisit the philosophical, mathematical or computational motivations and intuitions associated with $\st$ and elaborated in detail in the earlier literature on CoL.   

We call the old version of $\st$ and its dual $\cost$ {\bf tight}, and call the new versions of these operations {\bf loose}. Due to equivalence, the
purely  technical difference between the two versions does nor warrant introducing new symbols for the new operations. However, since this paper has to simultaneously deal with both versions, in order to avoid confusion, we shall use the symbols $\stt,\costt$ for the tight versions of $\st,\cost$, and the symbols $\stl,\costl$ for the loose versions. 

The intended audience for this relatively short (by the standards of CoL) and technical paper is expected to be familiar with the main concepts of CoL, such as those of static games, easy-play machines (EPM), the basic game operations, validity, and the related notions. If not, it would be both necessary and sufficient to read the first ten sections of \cite{Japfin}  for a self-contained, tutorial-style introduction. Having \cite{Japfin} at hand would probably be necessary in any case, because we  rely on the notation and terminology of \cite{Japfin} without reintroducing them, so any unfamiliar symbols or terms should be looked up in \cite{Japfin}, which has a convenient glossary\footnote{The glossary for the published version of \cite{Japfin} is given at the end of the {\em book} (rather than {\em article}), on pages 371-376. The reader may instead use the preprint version of \cite{Japfin}, available at http://arxiv.org/abs/cs.LO/0507045 The latter includes both the main text and the glossary.} for that.  The definition of $\stt$ given in \cite{Japfin} is longer than necessary for our purposes and, for that reason, the present paper reintroduces $\stt$ through a shorter definition. No other old operations or concepts will be reintroduced and, again, they are to be understood as defined or explained in \cite{Japfin}.   

\section{The two versions of branching recurrence}
Remember that,  in semiformal terms, a play of $\stt A$  starts as an ordinary play of game $A$. At any time, however, player $\bot$ (the environment) is allowed to make a ``replicative move'', which creates two copies of the current position $\Phi$ of $A$. From that point on, the game turns into two games played in parallel, each continuing 
from position $\Phi$. We use the bits $0$ and $1$ to denote those two threads, which have a common past (position $\Phi$) but possibly diverging futures. Again, at any time, $\bot$ can further branch either thread, creating two copies of the current position of that thread. If thread $0$ was branched, the resulting two threads will be denoted by $00$ and $01$; and if the branched thread was $1$, then the resulting threads will be denoted by $10$ and $11$. And so on: at any time, $\bot$ may split any of the existing threads $w$ into two threads $w0$ and $w1$. Each thread in the eventual run of the game will be thus denoted by a (possibly infinite) bitstring. The game is considered won by $\top$ (the machine) if it wins $A$ in each of the threads; otherwise the winner is $\bot$.

In formal terms, consider a constant game $A$.  There are two types of legal moves in (legal) positions  of $\stt A$: {\bf replicative} and  {\bf non-replicative}. To define these, we  agree that, where $\Phi$ is a position, by an {\bf node} of the underlying BT-structure\footnote{``BT'' stands for ``bitstring tree''.} of $\seq{\Phi}\stt A$
we mean a bitstring $w$ such that $w$ is either empty,\footnote{Intuitively, the empty bitstring is the name/address of the initial thread; all other threads will be descendants of that thread.} 
or else is $u0$ or $u1$ for some bitstring $u$ such that $\Phi$ contains the move $\col{u}$. 
Such a  node is said to be a {\bf leaf} iff it is not a proper prefix of any other node of the underlying BT-structure of $\seq{\Phi}\stt A$.\footnote{Intuitively, a leaf 
is the unique individual name of an already existing thread of a play over $A$, while a node $w$ which is not a leaf is 
a ``partial'' common name of several already existing threads --- namely, all threads whose individual names look like $wv$ for some bitstring $v$.}  A replicative move can only be made by (is only legal for) $\bot$, and such a move in a given position $\Phi$ should be $\col{w}$, where $w$ is a leaf of the underlying BT-structure of 
$\seq{\Phi}\stt  A$.\footnote{The intuitive meaning of move $\col{w}$ is splitting thread $w$ into two new threads $w0$ and $w1$.} As for non-replicative moves, they can be made by either player. Such a move by a player $\xx$ in a given position $\Phi$ should be $w.\alpha$, where   $w$ is a node of the underlying BT-structure of $\seq{\Phi}\stt A$ and $\alpha$ is a move such that, for any infinite extension $v$ of $w$, $\alpha$ is a legal move by $\xx$ in the position $\Phi^{\preceq v}$ of $A$.\footnote{The intuitive meaning of such a move $w.\alpha$ is making move $\alpha$ in thread $w$ and all of its (current or future) descendants.} Here and later, for a run $\Theta$ and a bitstring $x$, $\Theta^{\preceq x}$ means the result of deleting from $\Theta$ all moves except those that look like $u.\beta$ for some initial segment $u$ of $x$, and then further deleting the prefix ``$u.$'' from such moves.\footnote{Intuitively, $\Theta^{\preceq x}$ is the run of $A$ that has been played in thread $x$, if such a thread exists (has been  generated); otherwise, $\Theta^{\preceq x}$ is the run of $A$ that has been played in (the unique) existing thread which (whose name, that is) is
some initial segment of $x$.} A legal run $\Gamma$ of $\stt A$ is considered won by $\top$ iff, for every infinite bitstring $v$, $\Gamma^{\preceq v}$ is a $\top$-won run of $A$. This completes our definition of $\stt$. The dual operation $\costt$ is defined in a symmetric way, by interchanging $\top$ with $\bot$. That is, $\costt A=\gneg\stt\neg A$.

This was a brutally quick review, of course. See \cite{Japfin} for more explanations and illustrations. 

Anyway, now it is time to define $\stl$. 
A run $\Gamma$ is stipulated to be a legal run of $\stl A$ iff every move of $\Gamma$ has the prefix ``$w.$'' for some finite bitstring $w$ and,  
for any infinite bitstring $v$, $\Gamma^{\preceq v}$ is a legal run of $A$ (here $\Gamma^{\preceq v}$  means the same as before). Next, such a $\Gamma$  
is considered to be a $\top$-won run of $\stl A$ iff, for every infinite bitstring $v$, $\Gamma^{\preceq v}$ is a $\top$-won run of $A$. As always, 
the dual operation $\costl$ is defined in a symmetric way by interchanging $\top$ with $\bot$, or by stipulating that  $\costl A=\gneg \stl\gneg A$. 

Intuitively, $\stl A$ can be seen as parallel play of a continuum of threads/copies of $A$.\footnote{Nothing to worry 
about: ``playing a continuum of copies'' does not destroy the ``finitary'' or ``playable'' character of our games. Every move or position is still 
a finite object, and every infinite run is still an $\omega$-sequence of (lab)moves.} Each thread is denoted by an infinite bitstring and vice 
versa: every infinite bitstring denotes a thread. The meaning of a move $w.\alpha$, where $w$ is a finite bitstring, is making the move $\alpha$ simultaneously in all threads  of the form $wy$. Correspondingly, when $\Gamma$ is a legal run of $\stl A$  and $x$ 
is an infinite bitstring, $\Gamma^{\preceq x}$ represents the run of $A$ that took place in thread $x$.  And, in order to win the overall game $\stl A$, 
$\top$ needs 
to win $A$ in all threads. As we saw earlier, a similar characterization applies to $\stt A$ as well. However, the difference  --- again at the intuitive level --- is that, while in the tight version of the game the threads are generated/built/grown step-by-step through replicative moves (and ordinary moves of $A$ are only allowed to be made in existing threads), in the loose version all of the uncountably many threads are ``already there'' from the very beginning (which explains the absence of replicative moves), so that moves of $A$ can be made in any of them at any time. 

\section{The preservation of the static property}

Whenever a new game operation is introduced in CoL, one needs to make sure that it preserves the static property of games, for otherwise many things can go wrong. 

\begin{theorem}\label{jan31}
 The class of static games is closed under $\stl$ and $\costl$.
\end{theorem}

The rest of this section is devoted to a proof of the above theorem. Considering only $\stl$ is sufficient, because $\costl$ is expressible in terms of $\stl$ and $\gneg$, with $\gneg$ already known (Theorem 14.1 of \cite{Jap03}) to preserve the static property of games.  
\begin{lemma}\label{l14} 
Assume $A$ is a constant static game, $\Omega$ is a $\wp$-delay of $\Gamma$, and $\Omega$ is a $\wp$-illegal run of $\stl A$. Then $\Gamma$ is also a $\wp$-illegal run of $\stl A$.
\end{lemma}
\begin{proof} We will prove this lemma by induction on the length of the shortest illegal initial segment 
of $\Omega$. 
Assume the conditions of  the lemma. We want to show that $\Gamma$ is a $\wp$-illegal run of $\stl A$. Let $\seq{\Psi,\wp\alpha}$ be the shortest ($\wp$-) illegal initial segment of $\Omega$. Let $\seq{\Phi,\wp\alpha}$ be the shortest initial segment of $\Gamma$ containing all  $\wp$-labeled moves\footnote{In this context, different occurrences of the same labmove count as different labmoves. So, a more accurate phrasing would be ``as many $\wp$-labeled moves as...'' instead ``all  $\wp$-labeled moves of ...''.} of $\seq{\Psi,\wp\alpha}$. If $\Phi$ is a $\wp$-illegal position of $\stl A$
then so is $\Gamma$ and we are done. Therefore, for the rest of the proof, assume that 
\begin{equation}\label{654}
\mbox{\em $\Phi$ is not a $\wp$-illegal position of $\stl A$.}
\end{equation}

Let $\Theta$ be the sequence of those $\gneg\wp$-labeled moves of $\Psi$ that are not in $\Phi$. Obviously
\begin{equation}\label{e141}
\mbox{\em  $\seq{\Psi,\wp\alpha}$ is a $\wp$-delay of $\seq{\Phi,\wp\alpha,\Theta}$.}
\end{equation}
We also claim that
\begin{equation}\label{e142}
\mbox{\em $\Phi$ is a legal position of $\stl A$.}
\end{equation}
Indeed, suppose this was not the case. Then, by (\ref{654}), $\Phi$ should be $\gneg\wp$-illegal. This would make $\Gamma$  a $\gneg\wp$-illegal run of $\stl A$ with $\Phi$ as an illegal initial segment which is shorter than $\seq{\Psi,\wp\alpha}$. Then, by the induction hypothesis, any run for which $\Gamma$ is a $\gneg\wp$-delay, would be $\gneg\wp$-illegal. But, as observed in Lemma 4.6 of \cite{Jap03}, the fact that $\Omega$ is a $\wp$-delay of $\Gamma$ implies that $\Gamma$ is a $\gneg\wp$-delay of $\Omega$. So, $\Omega$ would be $\gneg\wp$-illegal, which is a contradiction because, according to our assumptions, $\Omega$ is $\wp$-illegal.

We are continuing our proof. There are two possible reasons to why $\seq{\Psi,\wp\alpha}$ is an illegal (while 
$\Psi$ being legal)  position of $\stl A$:

{\em Reason 1}: $\alpha$ does not have the form $w.\beta$ for some bitstring $w$ and move $\beta$. Then, in view of (\ref{e142}), 
$\seq{\Phi,\wp\alpha}$ is a $\wp$-illegal position of $\stl A$. As $\seq{\Phi,\wp\alpha}$ happens to be an initial segment of $\Gamma$, the latter then is a $\wp$-illegal run of $\stl A$, as desired.

{\em Reason 2}: $\alpha=w.\beta$ for some bitstring $w$ and move $\beta$ but, for some infinite extension $v$ of $w$, 
$\seq{\Psi,\wp\alpha}^{\preceq v}$ is an illegal --- and hence   $\wp$-illegal --- position of $A$. Clearly (\ref{e141}) implies that 
 $\seq{\Psi,\wp\alpha}^{\preceq v}$ is a $\wp$-delay of $\seq{\Phi,\wp\alpha,\Theta}^{\preceq v}$. Therefore, since 
$A$ is static, by Lemma 4.7 of \cite{Jap03}, $\seq{\Phi,\wp\alpha,\Theta}^{\preceq v}$ is a $\wp$-illegal position of $A$. But $\seq{\Phi,\wp\alpha,\Theta}^{\preceq v}=\seq{\Phi^{\preceq v},\wp\beta,\Theta^{\preceq v}}$.
A $\wp$-illegal position will remain illegal after removing a block of $\gneg\wp$-labeled moves (in particular, $\Theta^{\preceq v}$) at the end of it.
 Hence  $\seq{\Phi^{\preceq v},\wp\beta}$, which is the same as  $\seq{\Phi,\wp\alpha}^{\preceq v}$, is an illegal position of $A$. Consequently, 
$\seq{\Phi,\wp\alpha}$ is an illegal position of $\stl A$.
This, in view of (\ref{e142}), implies that  $\seq{\Phi,\wp\alpha}$ is in fact a $\wp$-illegal run of $\stl A$.
But then, as desired, so is $\Gamma$, because $\seq{\Phi,\wp\alpha}$  is an initial segment of it. 
\end{proof}

Assume  $A$ is a static constant game, $\Gamma$ is a $\wp$-won run of $\stl A$, and $\Omega$ is a $\wp$-delay of $\Gamma$. Our goal is to show that $\Omega$ is also a $\wp$-won run of $\stl A$ (this is exactly what $\stl A$'s being static means). 

If $\Omega$ is a $\gneg\wp$-illegal run of $\stl A$, then it is won by $\wp$ and we are done. So, assume that 
$\Omega$ is not $\gneg\wp$-illegal. According to the earlier mentioned Lemma 4.6 of \cite{Jap03}, if $\Omega$ is a $\wp$-delay of $\Gamma$, then $\Gamma$ is a $\gneg\wp$-delay of $\Omega$. So, by Lemma \ref{l14}, our $\Gamma$ cannot be $\gneg\wp$-illegal, for otherwise so would be $\Omega$. 
$\Gamma$ also cannot be $\wp$-illegal, because otherwise it would not be won by $\wp$. Consequently, $\Omega$ cannot be $\wp$-illegal either, for otherwise, by Lemma \ref{l14}, $\Gamma$ would be $\wp$-illegal. Thus, we have narrowed down our considerations to the case when both $\Gamma$ and $\Omega$ are legal runs of $\stl A$.

The fact that $\Gamma$ is a legal, $\wp$-won run of $\stl A$ implies that, for every (if $\wp=\top$) or some (if $\wp=\bot$)  infinite bitstring $v$,   
$\Gamma^{\preceq v}$ is a $\wp$-won run of $A$, and therefore (as $A$ is static and  $\Omega^{\preceq v}$ is obviously a $\wp$-delay of $\Gamma^{\preceq v}$) $\Omega^{\preceq v}$ is a $\wp$-won run of $A$. Since 
 $\Omega$ is a legal run of  $\stl A$, the above, in turn, means nothing but that $\Omega$ is a $\wp$-won run of $\stl A$. This completes our proof of Theorem \ref{jan31}.

\section{The equivalence between the two versions} 
\begin{theorem}\label{main} 
The formulas \ $\stt P\mli \stl P$ \  and \ $\stl P\mli \stt P$ \ are uniformly valid.
\end{theorem}

\begin{proof}  Uniform validity of $\stt P\mli \stl P$ means nothing but existence of an EPM ${\cal E}_1$ such that, for any static game $A$, ${\cal E}_1$  wins  
$\stt A\mli \stl A$, i.e. $\costt\gneg A\mld\stl A$. 
We define such an EPM/strategy/algorithm ${\cal E}_1$ as one that repeats the following routine over and over again (infinitely many times unless one of the iterations never terminates). At any step of the work of the algorithm, $\Psi$ stands for $\Phi^{1.}$, where $\Phi$ is the then-current position of the play. That is, $\Psi$ is the then-current position within the $\costt\gneg A$ component.  

\begin{quote} ROUTINE: Keep granting permission until the adversary makes a move $\beta$ that satisfies the conditions of one of the following two cases, and 
then act according to the corresponding prescription.

{\em Case 1}: $\beta$ is a move $w.\alpha$ in $\costt \gneg A$. Make the same move $w.\alpha$ in $\stl A$. 

{\em Case 2}: $\beta$ is a move $w.\alpha$ in $\stl A$. Make a series of replicative moves (if necessary) in $\costt \gneg A$ so that $w$ becomes a node of the underlying BT-structure of $\seq{\Psi}\costt \gneg A$. Then make the move $w.\alpha$ in $\costt \gneg A$.\end{quote}

Consider
any run $\Omega$ that could be generated when ${\cal E}_1$ (in the role of $\top$) plays as described. It is obvious that ${\cal E}_1$ does not make illegal moves unless its adversary does so first. So, if $\Omega$ is an illegal run of $\costt\gneg A\mld\stl A$, it is $\bot$-illegal and hence $\top$-won. Now assume $\Omega$ is a legal run of $\costt\gneg A\mld\stl A$. Let $\Sigma=\Omega^{1.}$ and $\Pi=\Omega^{2.}$. That is, $\Sigma$ is the run that took place in  $\costt\gneg A$, and $\Pi$ is the run that took place in $\stl A$. If, for every infinite bitstring $v$, $\Pi^{\preceq v}$ is a $\top$-won run of $A$, then $\top$ is the winner in the overall game because it is the winner in its $\stl A$ component. Suppose now $v$ is an infinite bitstring such that $\Pi^{\preceq v}$ is a $\bot$-won run of $A$. With a moment's thought, one can see that $\Sigma^{\preceq v}=\gneg\Pi^{\preceq v}$. So,   $\Sigma^{\preceq v}$ is a $\top$-won run of $\gneg A$. This makes $\Sigma$ a $\top$-won run of $\costt\gneg A$, and hence $\Omega$ a $\top$-won run of  $\costt\gneg A\mld\stl A$, as desired.\vspace{5pt}

To prove the uniform validity of $\stl P\mli \stt P$, we construct an  EPM ${\cal E}_2$ that wins $\costl\gneg A\mld\stt A$ for any static game $A$. The work of ${\cal E}_2$ consists in repeating the following routine over and over again. 
At any step of the work of the algorithm, $\Psi$ stands for $\Phi^{2.}$, where $\Phi$ is the then-current position of the play. That is, $\Psi$ is the then-current position within the $\stt A$ component. Also, ${\cal E}_2$  maintains a record $f$ for a mapping from the leaves $v$ of the underlying BT-structure of $\seq{\Psi}\stt A$ to finite bitstrings $f(v)$, such that (as can be easily seen from an analysis of the work of ${\cal E}_2$) 
\begin{equation}\label{feb4}
\mbox{\em for any two leaves $v_1\not=v_2$,  $f(v_1)$ is not a prefix of $f(v_2)$.}
\end{equation}
 At the beginning, the only leaf is $\epsilon$ (the empty bitstring), and the value of $f(\epsilon)$ is set to $\epsilon$.

\begin{quote}
ROUTINE: Keep granting permission until the adversary makes a move $\beta$ that satisfies the conditions of one of the following three cases, and then act according to the corresponding prescription.

{\em Case 1}: $\beta$ is a replicative move $\col{w}$ in $\stt A$. Let $v=f(w)$. Then  update $f$ by setting $f(w0)=v0$, $f(w1)=v1$ and without changing the value of $f$ on any other (old) leaves of the underlying BT-structure of $\seq{\Psi}\stt A$; do not make any moves.

 {\em Case 2}: $\beta$ is a non-replicative move $w.\alpha$ in $\stt A$. Let $u_1,\ldots,u_n$ be all leaves $u$ of the underlying BT-structure of $\seq{\Psi}\stt A$ such that $w$ is a prefix of $u$. And let $v_1=f(u_1),\ldots,v_n=f(u_n)$. Then make the moves $v_1.\alpha,\ldots,v_n.\alpha$ in $\costl\gneg A$; leave the value of $f$ unchanged.  

  {\em Case 3}: $\beta$ is a move $w.\alpha$ in $\costl \gneg A$. First assume there is a (unique due to (\ref{feb4})) leaf $x$ in the underlying BT-structure of $\seq{\Psi}\stt A$ such that  $w$ is a proper extension of $f(x)$. Then update $f$ by letting $f(x)=w$ and without changing the value of $f$ on any other  leaves; make the move $x.\alpha$ in $\stt A$. Now suppose there is no  leaf $x$ in the underlying BT-structure of $\seq{\Psi}\stt A$ such that  $w$ is a proper extension of $f(x)$. Let 
 $y_1,\ldots,y_n$ be all leaves $y$ of the underlying BT-structure of $\seq{\Psi}\stt A$ such that $w$ is a prefix of $f(y)$ (note: the set of such leaves well may be empty, i.e., $n$ may be $0$). Then make the moves $y_1.\alpha,\ldots,y_n.\alpha$ in $\stt A$; leave the value of $f$ unchanged. 
\end{quote}

Consider any run $\Omega$ that could be generated when ${\cal E}_2$  plays as described. As in the previous case, we may assume that $\Omega$ is legal, for otherwise it can be easily seen to be  $\bot$-illegal and hence $\top$-won.  Let $\Sigma=\Omega^{1.}$ and $\Pi=\Omega^{2.}$. That is, $\Sigma$ is the run that took place in  $\costl\gneg A$, and $\Pi$ is the run that took place in $\stt A$.  Further, for a number $i$ such that 
ROUTINE (in the scenario that generated $\Omega$) is iterated at least $i$ times, we let $f_i$ denote the value of the record $f$ at the beginning of the $i$'th iteration, and $\Psi_i$ denote the position reached by that time in the  $\stt A$ component.  

 Consider any infinite bitstring $v$ and assume that $\Pi^{\preceq v}$ is a $\bot$-won run of $A$ (if there is no such $v$, then obviously $\top$ is the winner in the overall game).  Let $z$ be an infinite bitstring satisfying the following condition:
\begin{quote} For any $i$ such that ROUTINE is iterated at least $i$ times, where $v_i$ is the (unique) prefix of $v$ such that $v_i$ is a leaf of the underlying BT-structure of $\seq{\Psi_i}\stt A$, we have that 
$f_i(v_i)$ is a prefix of $z$.\end{quote}
 With some analysis, details of which are left to the reader, one can see that such a $z$ exists, and that  
 $\Sigma^{\preceq z}=\gneg\Pi^{\preceq v}$. So,   $\Sigma^{\preceq z}$ is a $\top$-won run of $\gneg A$. This makes $\Sigma$ a $\top$-won run of $\costl\gneg A$, and hence $\Omega$ a $\top$-won run of  $\costl\gneg A\mld\stt A$, as desired.
\end{proof}

The present theorem can be applied to various particular $\st,\cost$-containing fragments of the (otherwise open-ended) language of CoL to show that
the two --- tight and loose --- understandings of $\st,\cost$ yield the  the same classes of valid or uniformly valid formulas. This would be done through a rather straightforward induction relying on the fact that the operators of the language respect equivalence in the sense of Theorem \ref{main}, and that modus ponens preserves validity and uniform validity. But, of course, Theorem \ref{main} establishes equivalence between the two versions of $\st,\cost$ in a much stronger sense than just in the sense of validating the same principles.

\end{document}